\documentclass[preprint,authoryear,12pt]{article}

\usepackage{newlfont}
\usepackage{a4wide}
\usepackage{amsmath}
\usepackage{inputenc}
\usepackage{amsfonts}
\usepackage{amssymb}
\usepackage{amsthm}
\usepackage{newlfont}
\usepackage{graphicx}
\usepackage{subfigure}
\usepackage{mathrsfs}
\usepackage[authoryear]{natbib}
\usepackage{color}
\usepackage[normalem]{ulem}
\usepackage{comment}
\newtheorem{theorem}{Theorem}
\newtheorem{lemma}{Lemma}

\newtheorem{definition}{Definition\rm}
\newtheorem{remark}{Remark}

\theoremstyle{definition}

\newcommand{\PP}{{\mathcal P}}
\renewcommand{\AA}{{\mathcal A}}
\newcommand{\MM}{{\mathcal M}}
\newcommand{\BB}{{\mathfrak B}}

\newcommand{\IR}{{\mathrm{I\!R}}}

\newcommand{\FF}{{\mathcal F}}

\newcommand{\EE}{{\mathcal E}}

\newcommand{\IN}{{\mathrm{I\!N}}}

\newcommand{\HH}{{\mathcal H}}

\newcommand{\TT}{{\mathcal T}}

\renewcommand{\SS}{{\mathscr S}}
\newcommand{\CC}{{\mathscr C}}

\inputencoding{ansinew}
\begin{document}

\title{An interim core for normal form games and exchange economies with incomplete information:  a correction.\footnote{This is an amended version of the original paper with the same title published in : \emph{Journal of Mathematical Economics 58(2015); 38-45.}}}

\author{Y. Askoura}

\maketitle

{LEMMA, Universit\'e Panth\'eon-Assas, Paris I\!I,\\
4 rue Blaise Desgoffe, 75006 Paris. email~: youcef.askoura@u-paris2.fr}

\begin{abstract} We consider the \emph{interim} core of normal form cooperative games and exchange economies with incomplete information based on the partition model. We develop a solution concept that we can situate roughly between Wilson's coarse core and Yannelis's private core. We investigate the \emph{interim} negotiation of contracts and address the two situations of contract delivery~: \emph{interim} and \emph{ex post}. Our solution differs from Wilson's concept because the measurability of strategies in our solution is postponed until the consumption date (assumed with respect to the information that will be known by the players at the consumption date). For \emph{interim} consumption, our concept differs from Yannelis's private core because players can negotiate conditional on proper common knowledge events in our solution, which strengthens the \emph{interim} aspect of the game, as we will illustrate with examples.  
\end{abstract}

\textbf{keywords :}
coarse core; fine core; private core; $\alpha$-core{; weak-core}; incomplete information; exchange economies; partition model.

JEL Classification codes : C02; C71.

\section{Introduction}  
We define and investigate an \emph{interim} core concept for normal form games with incomplete information. We focus on {the weak-core \citep{WEB81}, which is a slight modification of} the $\alpha$-core of normal form games \citep{AUM61,SCA71,KAJ92} and the core of exchange economies as initiated by \citet{RAD68}. The incomplete information aspect is modeled using Wilson's \citep{WIL78} partition model. Specifically, we associate to each player a $\sigma-$field representing the events that the particular player can discern.

\bigskip
\citet{WIL78} defined two concepts of the core that targeted two extreme situations: in the \emph{coarse core} concept, agents are not permitted to share their information, whereas in the \emph{fine core} concept, agents share all their information within coalitions. 
\citet{WIL78} obtained the non-emptiness of the coarse core by endowing the grand coalition with the particular power of sharing all its agents' information, while proper coalitions do not share any information and negotiate only over ``common knowledge'' events. The incoherence of this particular construction has been criticized in the literature, which has spawned alternative models.   

\bigskip
Several interesting studies have been undertaken based on Wilson's solution concepts. \citet{ALL96} provided an overview of the basic literature concerning exchange economies and cooperative games. \citet{YAN91} defined the private core, in which agents do not share their private information. Part of the incomplete information aspect of the private core and other types of models is represented by assuming that agents can only envision strategies that are measurable with respect to the $\sigma$-fields of events that they can discern, which amounts to assuming that contract delivery is situated at the \emph{interim} stage. \citet{KoY93} investigated the incentive compatibility of core concepts under incomplete information. The incentive compatibility problem arises when the prevailing state of nature is not publicly known prior to consumption or contract delivery. The investigation of the private core was further undertaken by \citet{GMY01} and \citet{AlY01}. \citet{PAG97} formulated a common treatment and proved the non-emptiness of the core for a unified model that can be reduced to some version of Wilson's coarse and fine cores\footnote{Indeed, \citet{PAG97} used \emph{ex ante} utilities, whereas in each of Wilson's coarse and fine cores, agents evaluate allocations based on possibly proper events that are smaller than the entire universe.} and Yannelis's private core as a function of information-sharing rule. \citet{SER01} extends the private core to dynamic economies and defines a non-myopic core. Additional information  on this subject, incentive problems and a cogent mathematical formulation of the problem of information sharing is provided in \citet{ALL06}. For a review of the alternative approach (in incomplete information games), Harsanyi's model \citep{HAR67}, we refer to the survey of \citet{FMV02}.

\bigskip
In this paper, we formulate an \emph{interim} core concept in which players negotiate with \emph{interim} utilities (conditional expectations) and can object to a status-quo allocation (strategy) for common knowledge events. The main difference between our concept and Wilson's concepts is that the measurability of strategies is not related to the information available at the negotiation date but to the information that will be known by the players at the consumption date. This modelling idea is more natural and lends consistency to the resulting concepts. Moreover, it solves the incoherence problem in Wilson's construction. From this perspective, players faces all circumstances that will be known just before contract delivery. In other words, they can envision strategies based on the information that will be received before consumption, and are not restricted to using strategies conceivable under the information known during the negotiation stage. Negotiation is conducted using the available information, but strategies can be based on all information that will be known-even future information-before consumption.  

If consumption is situated at the \emph{interim} stage (no additional information between negotiation and consumption), then our concept is close to Yannelis's private core, except that the negotiation is conducted on proper common knowledge events. Compared with the private core\footnote{This comment may be put in some perspective by the version of the private core defined in \citet{HaY01} and in \citet{CPY11} as a weak \emph{interim} private core. See point 2 in Section \ref{CPC} for a discussion of this concept~: in particular, it may fail to be nonempty under usual assumptions.}, this strengthens the \emph{interim} aspect of the model, as we will show with examples in Section \ref{CPC}. We address 
two situations, \emph{i.e.,} whether the prevailing state of nature is revealed at the consumption date or whether it is not. The weak-core of a normal form game and the core of an exchange economy are addressed in each case. Note that \citet{AST13} and \citet{MUT14} provide an \emph{ex ante} formulation of the $\alpha-$core, under Harsanyi's model.

\bigskip
This paper is divided into two main parts~: Section 2 is devoted to \emph{ex post} contract delivery and contains most of the notation and technical tools that we use. Section 3 is devoted to \emph{interim} contract delivery and concludes with a subsection comparing our concept to the private core. Section 4 contains a closing comment on information sharing. 

\section{\emph{Ex post} contract delivery}\label{EPCD}

In an incomplete information model, two important aspects must be considered~: (a) at what stage are negotiations made, \emph{ex ante} or \emph{interim}?\footnote{The \emph{ex post} situation corresponds to complete information.} and (b) at what stage is contract delivery made, \emph{ex ante, interim} or \emph{ex post}? For (a), this paper discusses the \emph{interim case}. For (b), two consumption dates, \emph{interim} and \emph{ex post}, are addressed. The \emph{interim} consumption is technically accounted for by the assumption of the measurability of players' strategies, which can be roughly addressed by acknowledging that players can only envision strategies that are measurable with respect to their respective $\sigma$-fields\footnote{Alternatively, new $\sigma-$fields are generally defined by the information sharing rule.}. Then, this assumption is related to conditions under which the game takes place. For instance, it is unnecessary if the prevailing state (or,  generally, a finer $\sigma-$field than all players' information $\sigma$-fields) is revealed before contract delivery, and it is necessary otherwise. In other words, if contract delivery is made at the \emph{ex post} stage,  the measurability condition can be relaxed, such as is the case in \citet{VOL00}, in which the agents' utilities are updated again following the information that agents possess at the negotiation step. \citet{VOL00} obtained an intermediate core concept between the coarse and the fine core of Wilson. In \citet{AKI12}, a new type of core (informational core) was introduced in which the measurability assumption is no longer required. In \citet{KOB80}, measurability was assumed with respect to a $\sigma$-field that will be revealed before contract delivery and a conditional core was introduced. In this section,  consumption is organized \emph{ex-post}.

\bigskip
\subsection{Conceptual aspects-Normal form games}\label{CANFG}
Let $(\Omega,\FF,\mu)$ be a probability space. $\Omega$ represents the set of states of nature. The probability $\mu$ is a common objective prior. In the sequel, two events with a $\mu$-null symmetric difference will be confused. $N=\{1,...,n\}$ is the set of players. If $S\subset N$ is a coalition, then we denote by $-S$ the coalition of the remaining players. The action space of player $i$ is denoted $A_i$. It is assumed to be a compact convex subset of a (separable) reflexive Banach space $X_i$. 
The $\sigma-$algebra $\BB(A_i)$ stands for the Borel $\sigma-$field of $A_i$.

Set $A=\prod_{i\in N}A_i$. The information of each player $i$ is represented by a sub-$\sigma$-algebra $\FF_i$ of $\FF$. The elements of $\FF_i$ represent the events that player $i$ can discern. In other words, for every $E\in \FF_i$, player $i$ knows whether the prevailing
state is in $E$ or in its complement, which is denoted $\complement E$. \citet{WIL78} defined the fine core to be the core concept corresponding to the situation in which agents within a coalition $S$ pool all their information. Then, they can discern all events in the coarsest sub-$\sigma$-algebra generated by $\underset{i\in S}{\cup}\FF_i$, which is denoted $\underset{i\in S}\vee \FF_i$. Analogously, Wilson's coarse core is that in which agents do not reveal their information within coalitions or use common knowledge events or events contained therein for a coalition $S$ in the field $\underset{i\in S}{\cap}\FF_i$, which is denoted $\underset{i\in S}{\wedge}\FF_i$.

\begin{itemize}
\item[\textbf{R1)}] Assume that there is a finite partition of $\Omega$ generating the $\sigma$-algebra $\underset{i\in N}{\vee}\FF_i$; then, every sub-$\sigma$-field $\FF_i$ is generated by a partition of $\Omega$. Denote by $\PP$ the coarsest partition generating the field $\underset{i\in N}{\vee}\FF_i$. Denote by $\PP_i$ the coarsest partition of $\Omega$ generating $\FF_i$. We can assume that $\mu(K)>0$ for all $K\in \PP_i$. The elements of $\PP_i$ are the finest events that can be discerned by player $i$ without sharing information with other players. Observe that for every $K\in \PP$, there is a unique $K_i\in \PP_i$ such that $K\subset K_i$.
\end{itemize}

Each agent $i\in N$ is assumed to know, before the negotiation date, the smallest (finest) event (an element of $\PP_i$) in his field $\FF_i$ containing the realized state of nature\footnote{This is only true in the case in which $\FF_i$ is generated by a partition of $\Omega$, as assumed above; otherwise, at this stage, player $i$ can only know whether the realized state is in $E$ or in $\complement E$, for all $E\in \FF_i$.}.

For each player $i$, associate a utility function $$u_i:A\times\Omega\rightarrow \IR_+$$ 

In the sequel, we assume that 
\begin{itemize}
	\item the information fields $\FF_i,i\in I,$ and the other components of the game are publicly known. 
\end{itemize}
For a player, knowing the whole structure of the information, that is, all the fields $\FF_i,i\in N,$ does not provide him with more information if he is only assumed to discern, in his own field, whether any event contains the prevailing state.

\bigskip
Player $i$'s strategy\footnote{Allocation in the case of exchange economies.}, $x_i$, is a function from $\Omega$ into his space of actions\footnote{Consumption set in the case of an exchange economy.}. The measurability condition requires that $x_i$ be measurable with respect to the $\sigma-$field containing the events that player $i$ can discern. For instance, this function must be $\FF_i$-measurable in the absence of  information sharing or, for example, $\underset{j\in S}{\vee}\FF_j$-measurable if $i$ belongs to the coalition $S$ and $S$ shares all its members' information. We consider the measurability condition by distinguishing two situations~:

\begin{itemize}
\item If the prevailing state of nature is not known before consumption takes place (contract delivery), then it is reasonable to assume the measurability condition. Without information sharing, this condition means that player $i$ plays a constant function in every set $K_i\in \PP_i$. As $i$ can only distinguish an event $K_i\in \PP_i$ containing the prevailing state and the consumption takes place at this step, \emph{i.e.}  before the exact prevailing state is revealed, $i$ plays a unique action, regardless of the exact realized $w\in K_i$, by seeking to maximize his mathematical expectation computed on $K_i$.
\item If the prevailing state of nature is known before consumption occurs, we can withdraw this measurability condition. Indeed, if each player knows the prevailing state at the consumption date, then he plays the action corresponding to the image of his strategy at this realized state\footnote{\label{WIFASS}Assuming that the whole information field $\FF$ allows precise knowledge of the exact realized state of nature.}.  
\end{itemize}

In this setting, the measurability condition is intimately related to the incentive compatibility problem. In other words, these two conditions are to be assumed or ignored simultaneously following the situation we consider.

In this section, we consider the situation in which 
\begin{itemize}
	\item[\textbf{C1)}] the realized state is known before consumption takes place$^8$. 
\end{itemize}

Thus, we do not require the incentive compatibility of solutions nor the measurability condition. 
\begin{itemize}
	\item[\textbf{C2)}] All strategies are assumed to be $\FF$-measurable.  
\end{itemize}

\begin{remark}If only the information carried by a $\sigma$-field $\HH$ such that $\FF\supset \HH\supset $ $\underset{i\in N}\vee \FF_i$ is revealed (the realized events of $\HH$) before consumption, then it is reasonable to assume the measurability of all strategies with respect to $\HH$. However, in this case, it is mathematically possible to modify the data of the game to withdraw this measurability condition (with respect to $\HH$ instead of $\FF$) by considering the game issued from the probability space $(\Omega,\HH,\mu_{|\HH})$.
\end{remark}

The game takes place as follows~:
At a given date, players receive their own private information that makes them able to recognize whether any set (only) in their own field contains the prevailing state. Following this date, players negotiate. They can form coalitions\footnote{The information-sharing role must be clarified in the negotiation process. Precisely, a player can either reveal his information to his coalition or not.} at this stage. At the end of this step, the true realized state is revealed. Thereafter, consumption (or contract delivery) takes place. 

\bigskip
The set of strategies $\AA_S$ of a coalition $S$ is an $|S|$-tuple of $\FF$-measurable functions $x_i:\Omega\rightarrow A_i,i\in S$. The set of strategies of the grand coalition $\AA_N$ is sometimes denoted simply by $\AA$. For each coalition $S\subseteq N$, $\AA_S$ is considered as a subset of $\underset{i\in S}{\prod} L_1(\Omega,\mu,X_i)$ that is endowed with the product of the weak topologies $\sigma(L_1(\mu,X_i),L_\infty(\mu,X^*_i))$. The $L_1$ spaces here stand for the spaces of Bochner integrable functions. The separability of the sets $A_i$ makes the strong measurability identical to Borel measurability (the image spaces $A_i$ are endowed with the Borel $\sigma$-field). For this reason we confuse these notions. Note that the sets $\AA_S$ are weakly compact \citep{DRS93,ULG91}.

\medskip
Associate to every player $i$ the utility $U_i$ defined on $\AA$ by $$U_i(x)_w=u_i(x(w),w)$$

Assume that

\begin{itemize}
\item for every $i$, $u_i$ is $\BB(A)\otimes \FF$-measurable and $\sup_{a\in A} |u_i(a,\cdot)|$ is integrably bounded.
\end{itemize}
\bigskip
Denote the described game by $$G=((\Omega,\FF,\mu),N,A_i,u_i,\FF_i)$$

In this paper, until section \ref{ISH}, we assume that 
\begin{itemize}
	\item [\textbf{C3)}] Information is not shared within coalitions. 
\end{itemize}

Then, every player $i$ must evaluate his expected payoff with respect to the conditional expectation defined for $x\in \AA$ by~: $$E(U_i(x)|\FF_i):\Omega \rightarrow \IR$$
Indeed, player $i$ can only know whether a particular event of $\FF_i$ may contain the prevailing state at the decision (negotiation) step.

\begin{definition}\label{ICOA} A coalition $S$ blocks a given strategy $x\in \AA$, if there exists a common-knowledge event $F\in \underset{i\in S}{\wedge}\FF_i$, {and if the exists $\varepsilon>0$} and if $S$ possesses an admissible strategy $y_S\in \AA_S$, such that for all $y_{-S}\in \AA_{-S}$,
$$ E(U_i(y_S,y_{-S})|\FF_i)_w>E(U_i(x)|\FF_i)_w{+\varepsilon},\text{ for a.e. } w\in F, \forall i\in S$$
In such a circumstance, we also posit that $S$ blocks $x$ on $F$.

The \emph{interim} core of $G$ is the set of non-blocked strategies.
\end{definition}

\begin{remark}
The concept defined above reduces to {an} \emph{ex ante} core if all of the information fields are trivial $\FF_i=\{\emptyset,\Omega\}$. Indeed, players evaluate strategies with respect to the mathematical expectation $E(U_i(x))=\int_\Omega u_i(x(w),w) d\mu$.
\end{remark}

{The previous definition is a generalization of the Weber's \citep{WEB81} weak-core which can be adapted for games with a finite set of players as follows. 
\begin{definition} \label{WC} Consider a game $(N,Y_i,h_i,{i\in N})$ where $N$ is a finite set of players, $Y_i$ the set of strategies of $i\in N$, and $h_i :Y=\prod_{i\in N}Y_i\rightarrow \IR$ the utility of $i\in N$. 
A coalition $S\subset N$ is said to block $y\in Y$ iff it possesses a strategy $x_S\in Y_S$, and there exists $\varepsilon>0$ such that, $$h_i(x_S,x_{-S})>h_i(y)+\varepsilon, \forall x_{-S}\in X_{-S},\forall i\in S.$$
The weak-core is the set of non blocked strategies.  
\end{definition}
By removing $\varepsilon$ in the previous definition, we obtain, the $\alpha$-core of \citet{AUM61}.
}
The non-emptiness of the \emph{interim} core for normal form games is shown in the following~:

\begin{theorem}\label{THMC} Assume that each $u_i$ is concave and upper-semicontinuous in its first argument, then the \emph{interim} core of $G$ is nonempty. 
\end{theorem}

\begin{proof} Following the idea of \citet{WIL78}, associate to $G$ a new game $G'$ defined as follows~:
the set of players is $(i,K), i\in N, K\in \PP_i$. Denote by $J$ the new set of indices. The set of strategies of $j=(i,K)$ is $Y_j=L_1(\Omega,\mu,A_i).\chi_K$, where $\chi_K$ is the characteristic function of $K$. Note that the sets $Y_j$ are weakly compact. Set $Y=\prod_{j\in J}Y_j$ and define the linear function (isomorphism) $L:Y\rightarrow \AA$ by $x=L(y)$, $x_i=\sum_{K\in \PP_i} y_{(i,K)}$. Associate to each player $j=(i,K)$ the utility $g_j$ defined for $y\in Y$ by $$g_j(y)=\int_{K} u_i(L(y)(w),w) d\mu=\mu(K) E(U_i(L(y))|\FF_i)_v, \forall v\in K$$

The set of admissible coalitions consists of subsets $S\subset J$ associated to pairs $(S_0,F)$ consisting of a coalition $S_0\subset N$ of the initial game and a common knowledge event $F\in \underset{k\in S_0}{\wedge} \FF_k$ as $(S_0,F)=\{(i,K):i\in S_0;K\in\PP_i\text{ and }K\subset F\}$. The {weak}-core of $G'$ is the set of strategies $y$ that are not blocked by an admissible coalition. It is clear that if $y$ is in the {weak}-core of $G'$, then $L(y)$ is in the \emph{interim} core of $G$. Note that in Definition \ref{ICOA}, the functions $E(U_i(y_S,y_{-S})|\FF_i)$ and $E(U_i(x)|\FF_i)$, $i\in S$, depend only on the restrictions of $(y_S,y_{-S})$ and $x$ to $F$, \emph{i.e.}, on $(y_S,y_{-S})_{|F}$ and $x_{|F}$.

We now check the routine regularity assumptions needed for the components of $G'$. Let us prove that for every $j=(i,K)\in J$, $g_j$ is upper semi-continuous and concave on $Y$. The concavity of $g_j$ results trivially from the concavity of $u_i$ in its first argument. Let us prove that $g_j$ is upper semi-continuous on $Y$.  Following the Eberlein-$\check{\text{S}}$mulian theorem, as we are working in a weakly compact space (note that a finite product of weakly compact sets is weakly compact), we can use sequences instead of nets to describe the closure of subsets and, consequently, to also describe the upper semi-continuity of real functions. Let $y^n$ converge weakly to $y$. Without loss of generality, consider a subsequence still denoted $y^n$ such that $\lim g_j(y^n)=\lim\sup g_j(y^n)$. The existence of the sup-bound of $g_j$ is ensured because $\sup_a|u_i(a,\cdot)|$ is integrably bounded.

Following \cite{DRS93}, there is a sequence $h_n\in co\{y^k:k\geq n\}$ such that $h^n$ converges a.e. to $y$, and hence $L(h^n)$ also converges a.e. to $L(y)$. Then, from the upper semicontinuity of $u_i$, for a.e. $w\in \Omega$,  $\lim\sup\; u_i(L(h^n)(w),w)\leq u_i(L(y)(w),w)$. Hence, by Fatou's Lemma~:  $$\lim\sup \;g_j(h^n)\leq \int_{K} \lim\sup\; u_i(L(h^n)(w),w)d\mu\leq \int_{K} u_i(L(y)(w),w)d\mu$$

From the concavity of $u_i$, for every $n\in \IN$, there is $k\geq n$ such that $g_j(h^n)\geq g_j(y^k)$. It follows that $$\lim \sup g_j(y^n)\leq g_j(y)$$

This proves the upper semi-continuity of the functions $g_j$.

Associate the characteristic function form game $G_C=(V,J)$, where the value function $V$ is defined by~:
$$
V(S)=\left\{
\begin{array}{l}
  v\in \IR^{J}:\text{ there is }y_S\in Y_S\text{ such that for all } y_{-S}\in Y_{-S}\\
   g_j(y_S,y_{-S})\geq v_j,\text{ for every } j\in S
\end{array}
\right\}
$$

on every coalition $S$ that is admissible for $G'$. For the other types of coalitions, simply let $V(S)=\IR_-^S\times \IR^{J\backslash S}$. 

It is obvious that to an element $v$ of the core of $G_C$ corresponds an element in the { weak-core} of $G'$. It is precisely an element $y\in Y$ such that the grand coalition {(admissible since it is associated to $(N,\Omega)$)} can ensure at least the gains $v_j$ at $y$ for its members, where $v$ is in the core of $G_C$. { Indeed, let $v$ in the core of $G_C$ and $y\in Y$, such that $g_j(y)\geq v_j$ for all $j\in J$. Hence, if an admissible coalition $S$ blocks $y$ for the blocking concept of definition \ref{WC}, then there exists $\varepsilon>0$ and $x_S\in Y_S$ such that $g_j(x_S,x_{-S})>g_j(y)+\varepsilon\geq v_j+\varepsilon$, for all $j\in S$ and all $x_{-S}\in Y_{-S}$.  By construction of $V(S)$, $v$ belongs to the interior of $V(S)$ and then, this is a contradiction with the fact that $v$ is in the core of $G_C$. Note that if $g_j$ does not depend on $y_{-j}$, for every $j\in J$, (this is the case of exchange economies in absence of externalities) then the previous argument works for the $\alpha-$core blocking concept, i.e., without the need of the $\varepsilon>0$.} 

For non-vacuity, Let us apply the result of \cite{SCA67} that guarantee the existence of an element in the core of $G_C$ under the following conditions~:
\begin{itemize}
\item[1)] $V(S)$ is closed and nonempty for every $S\subset J$ and $V(J)$ is bounded from above,
\item[2)] $V(S)$ is comprehensive for every $S\subset J$, \emph{i.e.}, if $v\in V(S)$ and $v'\leq v$ (coordinatewise), then $v'\in V(S)$,
\item[3)] $G_C$ is balanced, that is, for every balanced collection of coalitions $\SS$ with positive balancing weights $\delta_S,S\in \SS$,  one has $$ \cap_{S\in \SS}V(S)\subset V(J)$$
where the collection of coalitions $\SS$, with associated positive (balancing) weights $\delta_S,S\in \SS$, is said to be balanced precisely when $$\underset{S\in \SS,S\ni j}{\sum} \delta_S=1,\forall j\in J$$
\end{itemize}
The verification of 1) and 2) is easy. Indeed, 1) follows from the positivity, upper semi-continuity of the functions $g_j$ and the compactness of the strategy sets, and 2) is trivial from the definition of the sets $V(S)$. For 3), we apply an adapted Scarf trick \citep{SCA71}. Take a balanced collection $\SS$ with associated positive balancing weights $\delta_S,S\in \SS$, and $v\in \cap_{S\in \SS}V(S)$. For every $S\in\SS$, let $y^S\in Y_S$ such that $g_j(y^S,y^{-S})\geq v_j$ for every $j\in S$ and $y^{-S}\in Y_{-S}$. If $S$ is non-admissible, then $y^S$ can be taken arbitrarily because of the positivity of the functions $u_i$. 
Let the element $y=(y_1,...,y_{|J|})$ defined by $y_j=\underset{S\in \SS,S\ni j}{\sum} \delta_S y^S_j$. Then, $y\in Y$ by convexity and for any fixed $j\in J$, $y$ can be expressed as

$$y=\underset{S\in \CC,S\ni j}{\sum}\alpha_S(h_1^S,...,h_{|J|}^S) $$

where $h^S_m\in Y_m$ is defined by

$$h^S_m(w)=\left\{\begin{array}{ll}
        y^S_m(w), &\text{ if }m\in S, \\
        $\; $\\
       \displaystyle  \frac{\underset{C\ni m,C\not \ni j}{\sum}\delta_C y^C_m (w)   }{\underset{C\ni m,C\not \ni j}{\sum}\delta_C}, & \text{if }m\notin S.
      \end{array}\right.
 $$
where the last summations (and all of the following) are made, if not stated otherwise, over the coalitions $C$ (or $S$) belonging to $\SS$.

Indeed, for every $m$,
$$\underset{S\in \CC,S\ni j}{\sum}\delta_S  h_m^S=\underset{S\ni j,S\ni m}{\sum}\delta_S y_m^S+\underset{S\ni j,S\not \ni m}{\sum}\delta_S \frac{\underset{C\ni m,C\not \ni j}{\sum}\delta_C y^C_m}{\underset{C\ni m,C\not \ni j}{\sum}\delta_C}.$$
To conclude that the previous quantity yields $y_m$, it suffices to
remark that $$\underset{S\ni j,S\not \ni
m}{\sum}\delta_S=\underset{C\ni m,C\not \ni j}{\sum}\delta_C,$$
which is a consequence of the balancedness of the  collection of
coalitions~:
$$1=\underset{S\ni j}{\sum}\delta_S=\underset{S\ni j,S\not \ni m}{\sum}\delta_S+\underset{S\ni m,S\ni j}{\sum}\delta_S=\underset{C\ni m,C\not \ni j}{\sum}\delta_C+\underset{S\ni m,S\ni j}{\sum}\delta_S=\underset{C\ni m}{\sum}\delta_C=1.$$
Now with the help of the concavity of $g_j$, we obtain
$$
\begin{array}{rl}
g_j(y_1,...,y_{|J|})&\displaystyle=g_j(\underset{S\ni j}{\sum}\delta_S(h^S_1,...,h^S_{|J|})\\
                                     &\displaystyle=g_j(\underset{S\ni j}{\sum}\delta_S(y^S,h^S_{-S}))\\
                                     &\displaystyle\geq v_j
\end{array}
 $$

This establishes the balancedness of $G_C$, and consequently the non-vacuity of its core, and then that of the weak-core of $G'$. Take any $y$ in the weak-core of $G'$. Obviously, $L(y)$ is in the \emph{interim} core of the initial game.

\end{proof}

\subsection{\emph{Interim} core for exchange economies}\label{ICEE}

In this section, we apply the above \emph{interim} core concept for exchange economies. The difference with respect to normal form games is that we do not require externalities in player payoffs and the strategy sets are constrained. 

\bigskip
Most of the hypotheses considered above remain unchanged. We address only those new hypotheses related to exchange economies and, if not expressly stated, we use the previous conditions. 

For exchange economies, define the consumption set of player $i$ by $C_i=\IR_+^l$ for all $i\in N$, where $l$ is a positive integer. For every $i\in N$, let $e_i:\Omega \rightarrow C_i$ be the initial endowment of agent $i$. Assume that $e_i$ is $\FF$-measurable. Note that this condition implies that players do not precisely know their initial endowments at the \emph{interim} negotiation step. However, this information will be available at \emph{ex post} contract delivery.   

In contrast with the previous paragraph, allocations are subject to constraints. The set of admissible allocations $\CC_S$ of a coalition $S\subseteq N$ is an $|S|$-tuple of $\FF$-measurable functions $x_i:\Omega\rightarrow C_i,i\in S,$ such that $\sum_{i\in S} x_i(w)=\sum_{i\in S} e_i(w)$ for a.e $w\in \Omega$.  

Assume that
\begin{itemize}
\item the initial endowments $e_i$ belong respectively to $L_1(\Omega,\mu,C_i)$, \emph{i.e.}, they are $\mu$-integrable.
\end{itemize}

Denote by $P_i(\CC)$ the projection of $\CC=\CC_N$ on $L_1(\Omega,\mu,C_i)$ and note that $P_i(\CC)$ contains all of the projections of the sets $\CC_S$ on $L_1(\Omega,\mu,C_i)$. The $\mu$-integrability of the initial endowments makes the sets $\CC_S$ and the sets $P_i(\CC)$ weakly compact as we show in the following~:

\begin{lemma}\label{COMP} Under the condition of the integrability of $e_i,i\in N,$ the spaces $\CC_S$, $S\subset N$, and the projections $P_i(\CC)$ of $\CC$ on $L_1(\Omega,\mu,C_i)$ are weakly compact.
\end{lemma}

\begin{proof}
For a.e. $w\in \Omega$, for all $S\subseteq N$, all $x\in \CC_S$ satisfies $x_i(w)\in\{y\in \IR_+^l:0\leq y\leq \underset{j\in S}{\sum}e_j(w)\}$, and this last set is compact in $\IR^l$. Furthermore, 
$$\begin{array}{rl}
\int_E \|x\|_{|S|l}d\mu&= \underset{i\in S}{\sum} \int_E \|x_i\|_l d\mu\\
&\leq \underset{i\in S}{\sum} \int_E \|\underset{j\in N}{\sum}e_j\|_l d\mu\\ 
&\leq |S| \underset{j\in N}{\sum} \int_E \|e_j\|_l d\mu \overset{\mu(E)\rightarrow 0}{\longrightarrow}0
\end{array}
$$

where $\|.\|_{|S|l}$ stands for the sum of the Euclidean norms $\|.\|_{l}$ on the factor spaces $\IR^l$. That is, $\CC_S$ is pointwise bounded and uniformly integrable. As it is obviously weakly closed, it follows from Corollary 9 in \citet{ULG91} that it is weakly compact. The weak compactness of the sets $P_i(\CC)$ results from the continuity of the corresponding projections.
\end{proof}

For exchange economies, externalities are not considered, and the utility of player $i$ thus depends only on $x_i$ and $w$. To clarify a difference, denote it as $$\psi_i :C_i\times\Omega\rightarrow \IR_+$$
Assume that

\begin{itemize}
\item for every $i$, $\psi_i$ is $\BB(C_i)\otimes \FF$-measurable and $\sup_{a_i\in C_i} |\psi_i(a_i,\cdot)|$ is integrably bounded.
\end{itemize}

The other changes are analogous; thus another utility $\Psi_i$, defined on $L_1(\Omega,\mu,C_i)$ by $\Psi_i(x_i)_w=\psi_i(x_i(w),w)$, is associated to each player. In the absence of information sharing, the conditional expectation of player $i$ is defined by $E(\Psi_i(x_i)|\FF_i)$. 

Denote the described exchange economy by $$\EE=((\Omega,\FF,\mu),N,C_i,\psi_i,\FF_i,e_i)$$


{
\begin{definition}\label{CEXC} A coalition $S$ blocks a given strategy $x\in \CC$, if there exists a common-knowledge event $F\in \underset{i\in S}{\wedge}\FF_i$, and if $S$ possesses an admissible strategy $y^S\in \CC_S$, such that,
$$ E(\Psi_i(y^S_i)|\FF_i)_w>E(\Psi_i(x_i)|\FF_i)_w,\text{ for a.e. } w\in F, \forall i\in S$$
In such a circumstance, we also posit that $S$ blocks $x$ on $F$.
The \emph{interim} core of $\EE$ is the set of non-blocked strategies.
\end{definition}
}

\begin{theorem}\label{THMC_EC} Assume that each $\psi_i$ is concave and upper-semicontinuous in its first argument, then the \emph{interim} core of $\EE$ is nonempty. 
\end{theorem}
The proof is similar to that of Theorem \ref{THMC}, except for the following~: the initial endowments of the players $j=(i,K)$ in the new game (exchange economy), which we can denote here $\EE'$, defined at the beginning, are defined by $e_j=e_i\chi_K$. Set $Y_j=L_1(\Omega,\mu,C_i)\chi_K$. The set of admissible allocations for a coalition $S$ in $\EE'$ is defined by $\CC_S'=\{y^S\in \underset{j\in S}{\prod}Y_j : \underset{j\in S}{\sum} y^S_j=\underset{j\in S}{\sum}e_j \text{ a.e.}\}$. The weak compactness of these sets is shown in the previous Lemma. The upper semi-continuity of the functions $g_j$ may be shown analogously on the weak compact sets $P_j(\CC')$.   

The admissibility of the combination $y$ defined by $y_j=\underset{S\ni j}{\sum}\delta_S{y^S_j}$ (needed for the balancedness of the associated characteristic form game $G_C$) must be checked for the grand coalition as follows~:
observe first that for a.e. $w\in \Omega$, $\underset{j\in J}{\sum} e_j(w)=\underset{i\in N}{\sum} e_i(w). $ Then a.e. 
$$\underset{j\in J}{\sum} y_j =\underset{j\in J}{\sum} \;\underset{S\in \SS,S\ni j}{\sum}\delta_S y^S_j=\underset{S\in \SS}{\sum} \delta_S\underset{j\in S}{\sum} y^S_j=\underset{S\in \SS}{\sum} \delta_S \underset{j\in S}{\sum} e_j=\underset{j\in J}{\sum} \;\underset{S\in \SS,S\ni j}{\sum}\delta_S e_j=\underset{j\in J}{\sum}e_j$$

This formula establishes that $y$ is admissible. The fact that $g_j(y_j)\geq v_j$, for every $j$, results from the concavity of the functions $g_j$. The balancedness of $G_C$ results immediately. Indeed, thanks to the absence of externalities, we do not need the second calculation presenting, for $j$, $y=\underset{S\ni j}{\sum}\delta_S{h^S}$. {Thanks to the absence of externalities an element of the core of $G_C$ provides directly an element in the $\alpha-$core of $\EE'$ (there is no need to use the $\varepsilon$-approximation, as mentioned in the previous proof) and then an element in the interim core of $\EE$.}
For further details on this proof, please refer to the detailed proof of Theorem \ref{THMEED} in the next section.

\section{\emph{Interim} contract delivery}\label{ICD}

In this section, we consider the situation in which the realized state is not known prior to consumption. Our aim is to establish the results of the previous section in this setting. 

\vspace{12pt}
The game is organized as follows:
at a given date, players receive their own private information that makes them able to recognize whether any set, only in their own field, contains the prevailing state. After this date, players negotiate, and they can form coalitions at this step. After the negotiation step, but still in the \emph{interim} stage, contract delivery takes place. The true state is not revealed before consumption. Therefore, assume the

\begin{itemize}
\item Honest reporting hypothesis.
\end{itemize}

This hypothesis is required for contract delivery. It means that agents do not misreport their realized events to improve their utilities at contract delivery. As the forthcoming solution concept is close to the private core concept, let us underline that some versions of the private core are known to be incentive compatible (\citet{KoY93,HaY01}; see also \citet{HaY97} on this subject). Let us also underline that some new cores that have been recently defined \citep{CPY11} using ambiguity theory ingredients are also incentive compatible.  In this paper, we have adopted an objective common belief governing the states of nature. Then, (conditional) expected utilities seem more appropriate. 
Even if the above \emph{interim} core concept may be formulated successfully in the setting with ambiguity aversion, the techniques as elaborated thus far \citep{CPY11, AnK_DOI} cannot contribute, in our view, to the study of incentive problems in our case. Indeed, the tractability of the \emph{interim} focusing of the concept may be difficult even with these new tools. 
Aside from the existence of many incentive compatibility concepts, this problem may be difficult to address and may involve restrictive conditions (monotonicity and continuity). For now, assume that there is a well-informed authority that may guarantee the applicability of the negotiated contracts.

\vspace{12pt}
Assume all the hypotheses of subsection \ref{CANFG}, except C1) and C2). We instead assume the opposite of C1), and below we describe strategy sets as a substitution for C2).

\vspace{12pt}
The set of strategies $\MM_S$ of a coalition $S$ is an $|S|$-tuple of functions $x_i:\Omega\rightarrow A_i,i\in S,$ such that $x_i$ is $\FF_i$-measurable for each $i\in S$. The set of strategies of the grand coalition $\MM_N$ is occasionally denoted simply by $\MM$. For each coalition $S\subseteq N$, $\MM_S$ is considered to be a subset of $\underset{i\in S}{\prod} L_1(\Omega,\mu,A_i)$ endowed with the $L_1$ norm topology. Note that condition R1) on the existence of finite partitions generating player information fields together with the measurability of player strategies with respect to their sub-$\sigma$-fields render the sets of strategies of finite dimensions. Then, the weak and norm topologies coincide on these sets and make them obviously compact. Further recall that the measurability condition assumed on strategies means that~: without any other information regarding the realized state, player $i$ plays a constant function in each $K\in\PP_i$, which is a meaningful condition because the consumption occurs at the \emph{interim} stage. In fact, player $i$ knows only the smallest event $K\in \PP_i$ containing the prevailing state of the world.

Define the game just described in which consumption is \emph{interim} by $$G_I=((\Omega,\FF,\mu),N,A_i,u_i,\FF_i)$$

The blocking concept and the \emph{interim} core are defined similarly by replacing in definition \ref{ICOA}~: $\AA$ with $\MM$, $\AA_S$ with $\MM_S$, $\AA_{-S}$ with $\MM_{-S}$ and $G$ with $G_I$.

\begin{theorem}\label{THMC_M} Assume that each $u_i$ is concave and upper-semicontinuous in its first argument, then the \emph{interim} core of $G_I$ is nonempty. 
\end{theorem}

The proof is omitted because it is similar to that of theorem \ref{THMC}, except that the set of strategies $Y_j$, for $j=(i,K)$, must be set as $Y_j=A_i\chi_K$. 

\subsection{Exchange economies}\label{EE}

We consider here an exchange economy. Thus, consider the hypotheses of subsection \ref{ICEE}. As above, we address in the sequel only the necessary changes related to the measurability assumption.   

\begin{itemize}
	\item  $e_i:\Omega\rightarrow C_i$ is the initial endowment of player $i\in N$, $e_i$ is assumed to be $\FF_i$-measurable.   
\end{itemize}
This condition means that players know their initial endowments at the \emph{interim} stage. 

As in section \ref{ICEE}, no externalities are assumed on payoffs with respect to strategies; thus denote $\psi_i: C_i\times \Omega\rightarrow \IR_+$ as the utility of player $i$, and associate to $i$ the second utility~$\Psi_i$ analogously.

The set of allocations $\TT_S$ of a coalition $S$ is an $|S|$-tuple of functions $x_i:\Omega\rightarrow C_i,i\in S,$ such that $x_i$ is $\FF_i$-measurable for each $i\in S$ and $\sum_{i\in S} x_i(w)=\sum_{i\in S} e_i(w)$ for a.e. $w\in \Omega$. The set of allocations of the grand coalition $\TT_N$ is occasionally denoted simply by $\TT$. For each coalition $S\subseteq N$, $\TT_S$ is considered to be a subset of $\underset{i\in S}{\prod} L_1(\Omega,\mu,C_i)$ endowed with the $L_1$ norm topology.

Denote by $$\EE_I=((\Omega,\FF,\mu),N,C_i,\psi_i,\FF_i,e_i)$$ an exchange economy with incomplete information in which the consumption is organized at the \emph{interim} stage.

{
The blocking concept and the \emph{interim} core for $\EE_I$ are defined similarly by replacing, in the definition \ref{CEXC}, $\CC$ with $\TT$, $\CC_S$ with $\TT_S$, $\EE$ with $\EE_I$. }
\vspace{12pt}
\begin{theorem}\label{THMEED} Assume that each $\psi_i$ is concave and upper-semicontinuous in its first argument, then the \emph{interim} core of $\EE_I$ is nonempty.
\end{theorem}

As the proof of the analogous Theorem \ref{THMC_EC} is not detailed, we choose to provide the integral proof of this theorem for the reader's convenience.

\begin{proof}
Following the idea of \citet{WIL78}, we begin by associating a new game (economy), denoted $\EE'$, defined as follows~:
the set of players is $(i,K), i\in N, K\in \PP_i$. Denote by $J$ the new set of indices. The consumption set of $j=(i,K)$ is $Y_j=\IR^l_+\chi_K$, where $\chi_K$ is the characteristic function of $K$. Set $Y=\prod_{j\in J}Y_j$. Observe that $Y$ is of finite dimension, and consequently, the $L_1$ norm on the set of allocations of players $L_1(\Omega,\mu,C_i)$ induces a norm equivalent to the esssup norm  on the corresponding sets $Y_j$. The initial endowment of $j=(i,K)$ is $e'_j=e_i\chi_K$. The admissible set of allocations for a coalition $S$ is $\TT'_S=\{y_S=(y_j)_{j\in S}:\sum_{j\in S}y_j=\sum_{j\in S} e'_j \text{ a.e. and } y_j \in Y_j\}$.  Clearly, $\TT'_S$ is compact for every $S$; it is in fact of finite dimension, bounded and closed, as an intersection of a closed set with an inverse image of a singleton by a linear map.  Define the linear function $L:Y\rightarrow \prod_{i\in N} L_1(\Omega,\mu,C_i)$ by $x=L(y)$, $x_i=\sum_{K\in \PP_i} y_{(i,K)}$.
Observe that for a.e. $w\in \Omega$, $\underset{j\in J}{\sum} e'_j(w)=\underset{i\in N}{\sum} e_i(w)$ and that $y$ is admissible for the grand coalition $J$ in the new game iff $L(y)$ is admissible for $N$ in the initial game.

Associate to a player $j=(i,K)$ the utility $g_j$ defined for $y_j\in Y_j$ by $$g_j(y_j)=\int_{K} \psi_i(y_j(w),w) d\mu=\mu(K) E(\Psi_i(L(y)_i)|\FF_i)_v, \forall v\in K$$

The set of admissible coalitions consists of subsets $S\subset J$ associated to pairs $(S_0,F)$, where $S_0\subset N$ is a coalition of the economy $\EE_I$ and $F\in \underset{k\in S_0}{\wedge} \FF_k$ is a common knowledge event~: $S=\{(i,K):i\in S_0;K\in \PP_i \text{ and } K\subset F\}$.

Denote by $P_j(\TT')$ the projection of $\TT'=\TT'_N$ on $Y_j$, and note that $P_j(\TT')$ contains all of the projections of the sets $\TT'_S$ on $Y_j$, for every $S\ni j$. Clearly, $P_j(\TT')$ is compact.

Let us show that, for every $j=(i,K)\in J$, $g_j$ is upper semi-continuous and concave on $Y_j$. The concavity of $g_j$ results trivially from the concavity of $\psi_i$ in its first argument. For the upper semi-continuity of $g_j$, let $y_j^n$ converge to $y_j$. Then, $y_j^n$ converges a.e. to $y_j$. From the upper semi-continuity of $\psi_i$, for a.e. $w\in \Omega$,  $\lim\sup\; \psi_i(y^n_j(w),w)\leq \psi_i(y_j(w),w)$. Hence, by Fatou's Lemma  $$\lim\sup \;g_j(y^n_j)\leq \int_{K} \lim\sup\; \psi_i(y^n_j(w),w)d\mu\leq \int_{K} \psi_i(y_j(w),w)d\mu$$
The existence of the upper bound of $g_j$ is ensured by the fact that $\sup_a|\psi_i(a,\cdot)|$ is integrably bounded.
This proves the upper semi-continuity of the functions $g_j$.

Associate the characteristic function form game $G_C=(V,J)$, where the value function $V$ is defined by

$$V(S)=\left\{v\in \IR^{J}:\text{ there is }y^S\in \TT'_S\text{ such that }g_j(y^S_j)\geq v_j, \text{ for every } j\in S\right\}$$

on every admissible coalition $S$ of $\EE'$ (that is associated to a pair $(S_0,F)$, as described above). For the other form of coalitions, simply put $V(S)=\IR^S_-\times \IR^{J\backslash S}$.

It is obvious that to an element $v$ of the core of $G_C$ corresponds an element in the core of $\EE'$ (allocations not blocked by admissible coalitions). It is precisely an element $y\in \TT'$ such that the grand coalition can ensure at least the gains given by the corresponding components of $v$ at $y$ for its members, where $v$ is in the core of $G_C$ {(note that the grand coalition is always admissible, since it is associated to $(N,\Omega)$)}. Let us apply, for non-vacuity, the result of \cite{SCA67} that guarantee an element in the core of $G_C$ under conditions 1)-3) restated in the proof of Theorem \ref{THMC}. 

The verification of 1) and 2) is straightforward. Indeed, 1) follows from the positivity of $\psi_i$'s ($0\in V(S)$, for every $S$), the upper semi-continuity of the functions $g_j$ and the compactness of the admissible sets of allocations. Condition 2) results trivially from the definition of the sets $V(S)$. For 3), take a balanced collection $\SS$ with associated positive balancing weights $\delta_S,S\in \SS$, and $v\in \cap_{S\in \SS}V(S)$. Then, for every $S\in\SS$, let $y^S\in \TT'_S$ such that $g_j(y^S_j)\geq v_j$ for every $j\in S$. If $S$ is non-admissible, then $y^S$ can be taken arbitrarily. Let the element $y=(y_1,...,y_{|J|})$ be defined by $y_j=\underset{S\in \SS,S\ni j}{\sum} \delta_S y^S_j$. Let us check the admissibility of $y$, \emph{i.e.} $y\in \TT'$. We can write a.e.

$$\underset{j\in J}{\sum} y_j =\underset{j\in J}{\sum} \;\underset{S\in \SS,S\ni j}{\sum}\delta_S y^S_j=\underset{S\in \SS}{\sum} \delta_S\underset{j\in S}{\sum} y^S_j=\underset{S\in \SS}{\sum} \delta_S \underset{j\in S}{\sum} e'_j=\underset{j\in J}{\sum} \;\underset{S\in \SS,S\ni j}{\sum}\delta_S e'_j=\underset{j\in J}{\sum}e'_j$$

This formula establishes that $y\in \TT'$. The fact that $g_j(y_j)\geq v_j$, for every $j$, follows from the concavity of the functions $g_j$. The balancedness of $G_C$ results immediately. Consequently, the core of $G_C$ is nonempty and so is the core of $\EE'$. For any $y$ in the core of $\EE'$, $x=L(y)\in \TT$ is clearly in the \emph{interim} core of the initial economy $\EE_I$.
\end{proof}

\subsection{Comparison with \emph{interim} private cores}\label{CPC}
\bigskip
\begin{itemize}

\item[1.]
In some of the literature cited in the introduction, the welfare of the agents is evaluated by the mathematical expectation computed on the entire universe of states of nature even if the concepts addressed are \emph{interim}. In some others, as in the original definition of the private core, the evaluation is effected by means of a conditional expectation (an \emph{interim} utility) measured and considered, at the negotiation step, at every state of nature, including possibly commonly known unrealized events. In other words, to block an outcome, a coalition must improve upon its members' (\emph{interim}) utilities at every state of nature, including states of unrealized events. Ideally, in the \emph{interim} situation, the agents may know that proper events have been realized. Then, conditional expectations seem to be more appropriate and preferably be considered on proper realized events at the negotiation date. Unlike the private core, our concept defined above fulfils these criteria. To highlight this point, consider the situation of an economy $\EE_I$ described in Subsection \ref{EE}.
Let us make a comparison with the (\emph{interim}) private core in its original definition \citep{YAN91}, in which blocking is only possible over the entire $\Omega$. More precisely, $S$ blocks $x\in \TT$ iff there exists an allocation $y^S\in \TT_S$ such that $$E(\Psi(y^S_i)|\FF_i)_w> E(\Psi(x_i)|\FF_i)_w, \text{ for a.e. }w\in \Omega, \forall i\in S$$

Consider any economy as described above in which $\Omega$ is finite and assume that at one state of nature, say $w_0\in \Omega$, that all players discern as a non-null proper event, all the utilities are equal regardless of the manner in which players pool their initial endowments. Assume that negotiations are made at the \emph{interim} stage. If we adopt the definition of \citet{YAN91}\footnote{see also \citet{GMY01}, footnote p. 297, and \citet{AlY01}}, all issues of the game belong to the private core. Indeed, no coalition can improve upon the conditional expectations of its members at the fixed state $w_0$. Clearly, this is inappropriate if players can commonly discern another realized event not containing $w_0$ on which they can improve their utilities. This situation, extreme as it may seem, is solved by the \emph{interim} core defined above. Indeed, players can block allocations on proper events, and hence they may improve their expected utilities at least relative to the common knowledge event $\Omega \setminus \{w_0\}$.

\item[2.] The \emph{interim} private core of \citet{HaY01}, p. 499, and the weak \emph{interim} private core of \citet{CPY11}\footnote{In which it is mentioned that the existence of this issue is an open problem, p. 528.} are different from the concept defined above. To make a more suitable comparison, consider again the economy $\EE_I$, described in Subsection \ref{EE}.
To block an allocation, a coalition only has to improve its members' \emph{interim} (expected) utilities at a unique state of nature. Consider a game with three players $N=\{1,2,3\}$. Each player's information is given by a finest partition of $\Omega=\{a,b\}$, $\PP_1=\{\{a,b\}\}$,  $\PP_2=\{\{a\},\{b\}\}$ and $\PP_3=\PP_1$, respectively. Set $C_i=\IR_+$ and $e_i\equiv 1$, for all $i\in N$. Let the prior probability $\mu$ on $\Omega$ be the equiprobability $\mu(a)=\mu(b)$. The three utilities are given as follows~: $\psi_1(x_1,w)=x_1$, $\psi_3(x_3,w)=x_3$ for all $w\in \Omega, x_1,x_3\in \IR_+$ and $\psi_2(x_2,w)=x_2$ if $w=a$, and $\psi_2(x_2,w)=1-x_2$ otherwise. Let us show that the (weak) \emph{interim} private core of \citet{HaY01} and \citet{CPY11} is empty. For these core concepts, a coalition $S$ blocks an allocation $x$ if it possesses an admissible allocation $y\in \TT_S$, such that for a given state of nature $w_0\in \Omega$, one has $$ E(\Psi_i(y_i)|\FF_i)_{w_0}>E(\Psi_i(x_i)|\FF_i)_{w_0},\forall i\in S.$$
In such a circumstance, let us say in the sequel that $S$ blocks $x$ at $w_0$.
    
Given any allocation $x\in \TT$,
    $$
    \begin{array}{rl}
      E(\Psi_i(x_i)|\FF_i)_w= & x_i(w), \text{ if } i\in\{1,3\}, \text{ for all } w\in \Omega.\\
      E(\Psi_2(x_2)|\FF_2)_w= & \left\{\begin{array}{l}
                                    x_2(w) \text{, if } w=a,\\
                                    1-x_2(w) \text{, if } w=b.
                                    \end{array}
                                \right.
    \end{array}
    $$

Coalition \{i\}, for $i=1$ or $3$, blocks any allocation $x$, where $x_i(w)<1, w\in\{a,b\}$. Then, we analyze only allocations $x$, where $x_1 (w)\geq 1$ and $x_3 (w)\geq 1$ and assume that only $e_2$ can be shared among the three players to maximize their utilities. Such allocations can be written as $x=(1+\alpha_1,\alpha_2,1+\alpha_3)$, where $\alpha_1,\alpha_2$ and $\alpha_3$ are non-negative and satisfy $\sum_i\alpha_i\equiv 1$ and are freely determined by any coalition containing player 2.
Let $x=(1+\alpha_1,\alpha_2,1+\alpha_3)$ be an allocation satisfying the requirements detailed above. Note that the measurability condition implies that $x_i,i\in\{1,3\}$ is constant on $\{a,b\}$. We have only two possible cases~:  

1) for $i_0\in\{1,3\}$, $\alpha_{i_0}(a)>0$. Without loss of generality, assume that $i_0=1$. Then, the coalition $\{2,3\}$ blocks $x$ at $a$. A blocking allocation may be $y\in \TT_{\{2,3\}}$ given by $y_2=x_2+\alpha_1/2$ and $y_3=x_3+\alpha_1/2$. 

2) $x_1\equiv x_3\equiv 1$ then, $x_2\equiv 1$. In this case, $E(\Psi_2(x_2)|\FF_2)_b=0$ and for instance $\{1,2\}$ blocks $x$ at $b$ by playing the allocation $y\in \TT_{\{1,2\}}$ given by

$y_1\equiv x_1+1/2$ and $y_2\equiv 1/2$.       

This proves that the weak \emph{interim} private core is empty.  It is easy to check that the allocation $x=(1,1,1)$ belongs to the \emph{interim} core introduced in this paper. Note that the weak \emph{interim} private core discussed in this example is shown to be nonempty when stated under ambiguity aversion and complete preferences \citep{CPY11}. In particular, the measurability of player strategies is relaxed in this study, which, like the ambiguity setting, is incompatible with our \emph{interim} aspect and renders the defined core impossible to compare with our concept.       

\end{itemize}

\section{Comment on information sharing}\label{ISH}

One can elaborate, in our context, a core concept in which information can be shared. Let us illustrate this by providing details using the normal form game defined in section \ref{EPCD}. If information sharing is permitted, then we find it convenient that players evaluate their payoffs by means of the conditional expectation with respect to all known information at the negotiation date. Furthermore, it is natural to consider that players update their expected utilities relative to the information they possess at the negotiation step. However, the (direct) intuitive formulation of such a concept, which can be stated as follows, is problematic~:

\begin{definition}\label{IFIN1} A coalition $S$ blocks a given strategy $x\in \AA$, if there exists an event $F\in \underset{j\in S}{\vee}\FF_j$, {there exists $\varepsilon>0$}, and if $S$ has an admissible strategy $y_S\in \AA_S$, such that, for all $y_{-S}\in \AA_{-S}$
$$ E(U_i(y_S,y_{-S})|\underset{j\in S}{\vee}\FF_j)_w>E(U_i(x)|\underset{j\in S}{\vee}\FF_j)_w+\varepsilon,\text{ for a.e. } w\in F, \forall i\in S$$

The \emph{interim} fine core of $G$ is the set of non-blocked strategies.
\end{definition}

Indeed, if player $i$ can access the information carried by $\underset{i\in N}\vee \FF_i$ by joining the grand coalition, it is not meaningful to say that he participates in a proper coalition $S$ and evaluates his welfare by $E(U(.)|\underset{i\in S}\vee \FF_i)$. Indeed, he is assumed to discern the finer events contained in $\underset{i\in N}\vee \FF_i$, and then the mean of his payoffs on events of $\underset{i\in S}\vee \FF_i$ cannot matter to him. 

If we assume that each player intends to share his information with all or some of the other players, then the previous \emph{interim} core can be generalized to encompass this situation. Indeed, it suffices to set $H_i$ as the finest information that player $i$ may possess at the negotiation step. Then the \emph{interim} core, defined above, may be applied to this situation, by replacing $\FF_i$ with $H_i$ in definition \ref{ICOA}. For instance, a fine core concept can be defined as follows~:

\begin{definition}\label{IFIN2} A coalition $S$ blocks a given strategy $x\in \AA$, if there exists an event $F\in \underset{i\in S}{\wedge}H_i$, {there exists $\varepsilon>0$}, and if $S$ possesses an admissible strategy $y_S\in \AA_S$, such that, for all $y_{-S}\in \AA_{-S}$
$$ E(U_i(y_S,y_{-S})|H_i)_w>E(U_i(x)|H_i)_w+\varepsilon,\text{ for a.e. } w\in F, \forall i\in S$$

The \emph{interim} fine core of $G$ is the set of non-blocked strategies.
\end{definition}

By setting $H_i=\underset{j\in N}{\vee}\FF_j$, for all $i\in N$, we obtain a particular \emph{interim} core in which all the information is shared. Then, these \emph{interim} fine cores are nonempty following the theorems shown above. This definition seems to allow the formation of coalitions that can use the information of the remaining players to object to a status-quo allocation, which is an undesirable effect. However, because the defined concepts are nonempty, this can instead be seen as a stability property. That is, even if coalitions are tempted to behave in such a way, they cannot form by improving their members' payoffs.   

\section*{}
\section*{Acknowledgement}
The author is very grateful to an anonymous referee for his reading and suggestions that significantly improved the quality of this paper. 

\bibliographystyle{plainnat}
\bibliography{Interim_biblio}

\end{document}